\documentclass[conference,letterpaper]{IEEEtran}
\IEEEoverridecommandlockouts
\usepackage{cite}
\usepackage{amsmath,amsthm,amssymb,amsfonts}
\usepackage{algorithmic}
\usepackage{graphicx}
\usepackage{textcomp}
\usepackage{xcolor}
\usepackage{indentfirst}
\usepackage{fancyhdr}
\usepackage{float}
\usepackage{commath}
\usepackage[utf8]{inputenc}
\usepackage{amssymb}
\usepackage{enumerate}
\usepackage[shortlabels]{enumitem}
\usepackage[dvipsnames]{xcolor}
\usepackage{soul}
\usepackage{listings}
\usepackage{bm}
\usepackage{adjustbox}
\usepackage{cuted}
\usepackage{tikz}
\usepackage{pgfplots}
\usetikzlibrary{patterns}
\usepgfplotslibrary{fillbetween,colorbrewer}
\usetikzlibrary{positioning,shadows,shapes,backgrounds,calc,spy,arrows}
\usetikzlibrary{intersections}
\usepackage{pgfplotstable}
\pgfplotsset{compat=newest}
\usepackage{booktabs}
\usepackage{threeparttable}
\renewcommand{\arraystretch}{1.5}
\usepackage{caption}
\DeclareCaptionFont{smallcaps}{\scshape}  % Define small caps font
\usepackage{balance}
\newcommand{\Xset}{\mathcal{X}}

\newtheorem{theorem}{Theorem}

\theoremstyle{definition}
\newtheorem{example}{Example}

\renewcommand{\figurename}{Figure}
\renewcommand{\tablename}{Table}

\newcommand{\E}[1]{\mathbb{E}\left[#1\right]}

\newcommand{\thn}{\mathrm{th}}
\newcommand{\rin}{\mathrm{rin}}
\newcommand\notsotiny{\@setfontsize\notsotiny\@vipt\@viipt}
\newcommand{\Pcw}{\overline{P}_{\mathrm{cw}}}
\newcommand{\OMA}{\mathsf{OMA}}
\newcommand{\ER}{\mathsf{ER}}
\newcommand{\tx}{\mathrm{tx}}
\newcommand{\rx}{\mathrm{rx}}

\renewcommand{\d}{\mathrm{d}}

\DeclareMathAlphabet{\mathbbold}{U}{bbold}{m}{n}
\newcommand{\qmetric}{\mathbbold{q}}

\def\BibTeX{{\rm B\kern-.05em{\sc i\kern-.025em b}\kern-.08em
    T\kern-.1667em\lower.7ex\hbox{E}\kern-.125emX}}
    
\begin{document}

\title{Optical Communications with Relative Intensity Noise: Channel Modeling and Information Rates
\thanks{This research is part of the project COmplexity-COnstrained LIght-coherent optical links (COCOLI) funded by Holland High Tech  $|$ TKI HSTM via the PPS allowance scheme for public-private partnerships, and of the project AI-SUSAT with grant ID OAARE97526 of the research programme AiNed XS Europe financed by the Dutch Research Council (NWO).}
}

\author{\IEEEauthorblockN{Felipe Villenas}
\IEEEauthorblockA{\textit{Eindhoven University of Technology}\\
Eindhoven, The Netherlands\\
f.i.villenas.cortez@tue.nl}
\and
\IEEEauthorblockN{Yunus Can G\"ultekin}
\IEEEauthorblockA{\textit{Eindhoven University of Technology}\\
Eindhoven, The Netherlands\\
y.c.g.gultekin@tue.nl}
\and
\IEEEauthorblockN{Alex Alvarado}
\IEEEauthorblockA{\textit{Eindhoven University of Technology}\\
Eindhoven, The Netherlands\\
a.alvarado@tue.nl}
}

\IEEEaftertitletext{\vspace{-0.5em}}
\maketitle

\begin{abstract}
We consider optical communications with intensity modulation and direct detection affected by laser relative intensity noise (RIN). 
Starting from a continuous-time waveform model, we derive an equivalent discrete-time channel model. As a result of RIN, the resulting channel model exhibits signal-dependent noise with memory. Unlike the commonly-assumed model in the literature, the conditional variance of this noise term has a polynomial dependence on the symbol of interest. Finally, we study achievable information rates for this channel under practically-relevant system parameters. We take a mismatched decoding approach and compute the generalized mutual information (GMI) using a memoryless decoding metric. 
Our numerical results show that when the memory in the channel is ignored by the receiver, GMI saturates as the constellation size increases, and thus, dense constellations do not offer gains. We also show that this saturation results from nonsymmetric nonvanishing contributions of the symbols to the GMI.
\end{abstract}

% \begin{IEEEkeywords}
% Optical communications, intensity modulation and direct detection, relative intensity noise, achievable information rates, mismatched decoding, generalized mutual information. %fiber optics,
% \end{IEEEkeywords}

\section{Introduction}

%\lipsum[1-5]

\IEEEPARstart{F}{ueled} by the exploding demand from artificial intelligence applications, the need for increasing data rates in short-reach data center interconnects (DCI) is unstoppable \cite{zhou2019beyond}. Maintaining relatively low hardware costs is also of key importance, and thus, the vast majority of these optical links employ intensity modulation (IM) and direct detection (DD) where only the amplitude of the transmitted optical light is modulated to carry information~\cite{che2023modulation}. Pulse amplitude modulation (PAM) is typically employed in these systems, enabling low-complexity transceivers with low power consumption \cite{wei2020experimental}.

Rates up to 200 Gbps have been shown to be feasible for IM-DD systems with PAM-4~\cite{baseline_200G}. To achieve 400 Gbps, doubling the symbol rate is the most straightforward option. However, increasing the symbol rate also means integrating the noise power spectral density (PSD) over a larger bandwidth. Thus, the transmit power must also increase to maintain or improve the signal-to-noise ratio. This increase in signal power leads to the generation of signal-dependent noise caused by relative intensity noise (RIN) of the laser. 

RIN is typically measured from a constant light source (e.g., a laser), without an information-carrying signal \cite{agrawal2013semiconductor,RIN_Keysight}. RIN is then modeled as memoryless with power proportional to the square of the instantaneous light intensity \cite{agrawal2013semiconductor}. In PAM systems, this assumption on RIN results in a model widely used in the literature (see, e.g., \cite{szczerba20124,rizzelli2023analytical,maneekut2024reach}), where the noise power is proportional to the square of the transmitted symbol.  
In this paper, we present a detailed analysis of RIN-dominated high-speed IM-DD systems. 
Starting from the continuous-time transfer characteristics of the signal processing blocks in the transceiver chain, we develop a discrete-time equivalent channel model that explicitly accounts for the statistics of the transmitted symbol sequence. 
The resulting model demonstrates that the commonly used signal-dependent noise model in the literature is imprecise, as it computes the noise statistics neglecting the properties of the information-carrying signal.
% In this paper, we take a detailed look at RIN-dominated high-speed IM-DD systems.
% Starting from continuous-time transfer characteristics of signal processing blocks in the transceiver chain, we develop a discrete-time equivalent channel model, now also considering the statistics of the symbol sequence.
% The resulting model we develop in this paper reveals that the model used in the literature is imprecise, since the statistics of the signal-dependent noise are computed neglecting the properties of the information-carrying signal. 

The contributions of this paper are threefold.
Firstly, we develop a novel analytical  discrete-time channel model for such systems. The proposed model---unlike the one often used in the literature\cite{szczerba20124,rizzelli2023analytical,maneekut2024reach}---is a channel with memory. 
Secondly, we show through numerical simulations that the proposed model accurately characterizes the signal-dependent noise for arbitrary continuous time PAM waveforms generated with band-limited pulses.
Lastly, we evaluate achievable information rates (AIRs) of the proposed channel model taking a mismatched decoding perspective, where the decoding metric is assumed to be memoryless. Our AIR results show that for typical practical system parameters, using $M$-PAM modulation formats with more than $8$ points results in negligible gains, unless the channel memory is taking into account.

The paper is organized as follows.
Sec.~\ref{sec:system_model} describes the system model and derives the novel discrete-time channel model with memory.
The statistical properties of the noise and memory are investigated in Sec.~\ref{sec:noisestat}.
AIRs are studied in Sec.~\ref{sec:air} before the conclusions are drawn in Sec.~\ref{sec:conclusions}.

%In IM-DD optical links, the increase in signal power leads to signal-dependent noise \cite{szczerba20124}. Without optical amplification, relevant signal-dependent noise sources to take into consideration are the shot noise generated at the receiver \cite{safari2015efficient,van2018optimization}, and the relative intensity noise (RIN) generated from the transmitter laser \cite{baveja201756,villenas2025ecoc}. The latter is the dominant noise source in the high optical power regime \cite[Sec.~II]{szczerba20124}. Consequently, one of the main challenges for achieving 400 Gb/s/lane and beyond, is to address signal-dependent noise. Due to the reasons above, here we focus on IM-DD systems without optical amplification subject to the presence of RIN. 
\section{System and Channel Model}
\label{sec:system_model}

\begin{figure*}[!t]
\vspace{0.3cm}
    \centering
    \begin{adjustbox}{width=0.9\linewidth,trim={1.95cm 0 0.25cm 0}, clip}
        \definecolor{ashgrey}{rgb}{0.75, 0.75, 0.75}
\definecolor{antiquebrass}{rgb}{0.98, 0.81, 0.69}
\definecolor{brilliantlavender}{rgb}{0.96, 0.73, 1.0}
\definecolor{bgcolor}{rgb}{0.97, 0.98, 0.96}
\definecolor{bgcolor2}{rgb}{0.99, 0.93, 0.89}
\newcommand{\OColor}{blue!80!black}

\tikzstyle{block} = [draw, line width = 0.6pt, fill=black!20, rectangle, minimum height=30pt, rounded corners=0.1cm, text width=2.5em,align=center]

\tikzstyle{block_wide} = [draw, line width = 0.6pt, fill=black!20, rectangle, minimum height=30pt, rounded corners=0.1cm, text width=3.5em,align=center]

\tikzstyle{block_wide2} = [draw, line width = 0.6pt, fill=black!20, rectangle, minimum height=20pt, rounded corners=0.1cm, text width=3em,align=center]

\tikzstyle{block_wide3} = [draw, line width=0.6pt, fill=black!20, rectangle, minimum height=40pt, rounded corners=0.1cm, text width=7.0em,align=center]

\tikzstyle{block2} = [draw, line width = 0.6pt, fill=black!20, rectangle, minimum height=30pt, minimum width=30pt, rounded corners=0.1cm, text width=5em,align=center]

\tikzstyle{Cir} = [draw, circle,  minimum size=2.15em]
\tikzstyle{circ} = [circle, draw, minimum size=10pt, text centered,inner sep=0pt]

\tikzstyle{gain} = [draw, line width=0.5pt, fill=gray!20, isosceles triangle, isosceles triangle apex angle=60, minimum height=25pt, minimum width=20pt, text width=1em, align=center]

\newcommand{\shotNoise}{0}  % Include shot noise
\newcommand{\lw}{0.7}       % Line width

\begin{tikzpicture}[auto, node distance=1 cm,>=to,line width=\lw]

    %%%% IM/DD
    %% placing the blocks
    % TX 
    \node [coordinate] (input) {};  
    %\node [block, right = 1.5em of input, fill = gray!30] (DAC) {$\!f^{-1}(\cdot)$};
    %\node [block, right = 2em of input, fill = gray!30] (opticalPAM) {$\eta(\cdot+\beta)$};
    \node [block, right = 2em of input, fill = gray!30] (pulse) {Filter $p(t)$};
    \node [block, right = 2.5em of pulse, fill = gray!30] (DAC) {Pre Dist.};
    %\node[circ,right=2.5em of pulse,xshift=0pt](DAC){\large $\times$};
    %\node [below=1.3em of DAC](tx_gain){$1/\Pcw$};
    \node [block,right = 1.5em of DAC, fill = red!20] (EO) {MZM};
    %\node [below=0em of EO](){$f_{\mathrm{MZM}}(\cdot)$};

    % --- laser block ---
    \node[block_wide3,dashed,above=1.1em of DAC,xshift=-20pt,fill=red!20] (laser) {};
    \node[left=0.7em of laser,rotate=90,xshift=1.5em](){Laser};
    \node[circ,above=1.5em of DAC,xshift=6pt](sum1){\large $+$};
    \node[left=1.5em of sum1](P_A){$\Pcw$};
    \node[above=0.2em of P_A](P_A_rin){$\Pcw\!\cdot\!N_{\mathrm{rin}}(t)$};
    % \node [block_wide2,above = 1.5em of EO, fill = brilliantlavender] (Laser) {Laser};
    % \node[above=0em of Laser](){RIN};

    % Fiber
    \draw [\OColor,solid,line width=0.5pt,opacity=1] ($(EO.east)+(2.3,0.3)$) circle (3mm);
    \draw [\OColor,solid,line width=0.5pt,opacity=1] ($(EO.east)+(2.4,0.3)$) circle (3mm);
    \draw [\OColor,solid,line width=0.5pt,opacity=1] ($(EO.east)+(2.5,0.3)$) circle (3mm);
    \node[right=1.85cm of EO, yshift=-1.8em](fiber_mid){($L$ km)};
    \node[right=0.7em of EO, yshift=0.8em](){$P_{\mathrm{tx}}(t)$};

    % % RX
    %\node [block_wide, right = 5 cm of EO, fill = green!20] (OE) {PD-TIA};
    \node [block, right = 5cm of EO, fill = red!20] (OE) {PD};
    \node[left=0.9em of OE, yshift=0.8em](Prx){$P_{\mathrm{rx}}(t)$};
    \node[below=0em of OE](){$\mathfrak{R}$};
    
    \node[circ,right=2.5em of OE](sum3){\large $+$};
    \draw [draw,-latex] (OE) -- node[midway,above]{$I(t)$} (sum3);

    \node [above=1.8em of sum3](noise){$N_{\thn}(t)$};
    \node [circ, right = 1.5em of sum3] (TIA) {\large $\times$};
    \node [below=1.3em of TIA](rx_gain){$G$};
    \node [block, right = 2.5em of TIA, fill = gray!30] (ADC) {Filter $h(t)$};
    \node [below=0em of ADC](){};
    \node [right=0.1em of ADC,yshift=8pt](){$Y(t)$};
    \node [right=2.5em of ADC] (sampling) {};
    \node [below=0.4em of sampling, xshift=0.8em] () {\small{$t=kT$}};
    \node [right = 0.8em of sampling] (output) {};

    % % We draw edges between nodes
    \draw [draw,-latex] (input) -- node[left,text width=2.2em, align = left]{$X_k$}(pulse);
    % \draw [draw,-latex] (input) -- node[left]{$X$} (OOK);
    %\draw [draw,-latex] (opticalPAM) -- node[midway,above]{$A_k$}(pulse);
    \draw [draw,-latex] (pulse) -- node[midway,above]{$S(t)$}(DAC);
    \draw [draw,-latex] (DAC) -- node[midway,above]{}(EO);
    %\draw [draw,-latex] (tx_gain) -- (DAC);
    \draw [draw,-latex] (rx_gain) -- (TIA);
    
    %\draw [draw,-latex, color=blue] (Laser) -- (EO);
    \draw [draw,-latex, color=\OColor] (sum1) -| node[midway,right,xshift=-3.0em,yshift=0.8em]{\textcolor{black}{$P_{\mathrm{cw}}(t)$}}(EO);
    \draw [draw,-latex, color=\OColor] ($(P_A)+(1.0em,0)$) -- (sum1);
    %\draw [draw,-latex, color=\OColor] ($(P_A_rin)+(0,-0.6em)$) -- (sum1);
    \draw [draw,-latex, color=\OColor] ($(P_A_rin.east)+(-2pt,0)$) -| (sum1);
    
    \draw [draw,-latex, color=\OColor] (EO) -- node[midway,below]{\textcolor{black}{Fiber}}(OE);
    \draw [draw,-latex] (noise) -- (sum3);

    \draw [draw,-latex] (sum3) -- (TIA);
    \draw [draw,-latex] (TIA) -- node[midway,above]{$V(t)$}(ADC);
    % \draw [draw,-latex] (INTER2) -- node[midway,above]{$y_k$} (HD);
    \draw[draw] (ADC) -- (sampling);
    \draw [draw] ($(sampling.west)+(0.2em,1em)$) -- ($(output.west)$);
    \draw [draw,-latex] ($(output.west)$) --  node[right,text width=2.4em, align = right]{$Y_k$}($(output.west)+(3em,0)$);

    % sampling arrow
    \draw[->, >=latex, densely dashed, gray, bend right=20] ($(sampling)+(0.6em,1.2em)$) to ($(sampling)+(0.6em,-0.8em)$);

    % Option (1): half-rectangle and channel on top on bottom line      

    \node[below=5em of Prx,xshift=-1.2em](ch_eq){Theorem~\ref{theo:chanmod}: $Y_k = X_k + Z_k + Q_k$};%{$Y_k=X_k+Z_k\sqrt{\sigma_{\mathrm{ele}}^2 + (X_k+\beta)^2\sigma_{\mathrm{rin}}^2}$};
    \draw[draw,-stealth,dotted,thick,color=gray] ($(pulse.south)+(-1.05,1.5em)$) |- (ch_eq);
    \draw[draw,-stealth,dotted,thick,color=gray] ($(ADC.south)+(7.8em,1.5em)$) |- (ch_eq);

    % Legend
    \node[above = 2.075cm of input, xshift=+16.0em] (elec_0){};
    \node[above = 2.075cm of input, xshift=18.5em] (elec_1){};
    \node[below = 0.4em of elec_0, xshift=0em] (opt_0){};
    \node[below = 0.4em of elec_1, xshift=0em] (opt_1){};
    
    \draw[draw,-latex] (elec_0) -- node[right,text width=4.3em, align=left,xshift=0.6em]{\footnotesize Electrical}(elec_1);
    \draw[draw,-latex,\OColor] (opt_0) -- node[right,text width=3.5em, align=left,xshift=0.6em]{\footnotesize \textcolor{black}{Optical}}(opt_1);
    
    \begin{pgfonlayer}{background}
        \draw[dotted,fill=bgcolor,rounded corners=2pt] ($(pulse.west)+(-0.25,-1.2)$) rectangle ($(EO.east)+(+0.15,2.5)$);     % Tx
        \draw[dotted,fill=bgcolor,rounded corners=2pt] ($(EO.east)+(+0.30,-1.2)$) rectangle ($(OE.west)+(-0.40,1.6)$);     % Fiber
        \draw[dotted,fill=bgcolor,rounded corners=2pt] ($(OE.west)+(-0.25,-1.2)$) rectangle ($(ADC.east)+(+1.9,1.6)$);     % Rx

        \node[below=1.1em of fiber_mid](optical_text){Optical Channel};
        \node[left=7em of optical_text](){Transmitter};
        \node[right=12em of optical_text](){Receiver};
        
    \end{pgfonlayer}
    
    % Option (2): full-rectangle and channel below bottom line 
    % \draw[dashed,thick,color=gray] ($(DAC.south)+(-0.95,8em)$) -- ($(DAC.south)+(-0.95,-1.6em)$);
    % \draw[dashed,thick,color=gray] ($(DAC.south)+(-0.95,8em)$) -- ($(ADC.south)+(0.95,8em)$);
    % \draw[dashed,thick,color=gray] ($(ADC.south)+(0.95,8em)$) -- ($(ADC.south)+(0.95,-1.6em)$);
    % \draw[dashed,thick,color=gray] ($(DAC.south)+(-0.95,-1.6em)$) -- ($(ADC.south)+(0.95,-1.6em)$);
    % \node[right=1.3cm of DAC, yshift=-4.1em](){$Y=X+Z\sqrt{\sigma_{\mathrm{thn}}^2 + X^2\sigma_{\mathrm{rin}}^2}$};
 
\end{tikzpicture}
    \end{adjustbox}
    \caption{Equivalent high-speed IM-DD system with laser RIN and external modulation.}
    \label{fig:system_IMDD}
    \vspace{-0.5em}
\end{figure*}

We consider the baseband equivalent IM-DD system shown in Fig.~\ref{fig:system_IMDD}. 
We denote the $M$-PAM constellation by the random variable (RV) $X\in\Xset$, with $X>0$ and $|\Xset|=M$. The outermost symbols of the constellation are determined by the optical modulation amplitude (OMA) and extinction ratio (ER) parameters, which are defined as
\begin{align}
    %\OMA \triangleq \mathsf{a}_M-\mathsf{a}_1\; \text{and}\;\ER \triangleq \mathsf{a}_M/\mathsf{a}_1.\label{eq:OMA}
    \OMA \triangleq \max\{\Xset\}-\min\{\Xset\}\; \text{and}\;\ER \triangleq \frac{\max\{\Xset\}}{\min\{\Xset\}}.\label{eq:OMA}
\end{align}
A sequence of PAM symbols $\{X_n\}_{n\in\mathbb{Z}}$ is transmitted, where the elements $X_n$ are independent and uniformly distributed (i.u.d.) RVs. Pulse shaping with a filter impulse response $p(t)$ is used. At the output of the transmitter filter, the generated waveform is
\begin{equation}
    S(t) = \sum_{n\in\mathbb{Z}}X_n p(t-nT),
\end{equation}
where $T$ is the symbol period defined by the signaling rate $R_\mathrm{s} = 1/T$. The filter $p(t)$, normalized to unit amplitude $p(0)=1$, is an arbitrary pulse shape that, in combination with the filter $h(t)$ at the receiver (see Fig.~\ref{fig:system_IMDD}) forms a combination that satisfies the zero intersymbol interference (ISI) Nyquist criterion, e.g., both can be root-raised-cosine (RRC) filters. We assume that \mbox{$\ER<\infty$} and that $\min\{\Xset\}$ is selected such that \mbox{$\mathrm{Pr}\{S(t)<0\}<\epsilon$} where a typical value for $\epsilon=10^{-5}$ \cite[Sec.~V-A]{che2021does}. This assumption ensures that $S(t)$ is virtually nonnegative.\footnote{The nonnegativity is imposed by the electro-optical modulator for its input. In practice, negative values at the modulator input are folded onto positive values by the modulator transfer function. For our analysis here with $\epsilon=10^{-5}$, we assume that this never happens.}

% Furthermore, for a constellation $\mathcal{A}$ defined by \mbox{$\ER<\infty$},\footnote{The case when $\ER \to \infty$ happens when $\min\{\mathcal{A}\}=0$ and may only hold when using rectangular pulses with infinite bandwidth.} the pulse $p(t)$ must satisfy (i) $p(0)=1$, and (ii) \mbox{$\mathrm{Pr}\{S(t)<0\}<\epsilon$}. For (root)-raised-cosine pulses, the roll-off factor may enhance the peak of the waveform $S(t)$. Thus, in order to fulfill (ii) the maximum $\ER$ as a function of the roll-off factor is shown in Fig.~\ref{fig:rolloff_ER}. The maximum $\ER$ value is obtained by using the cumulative distribution function of $S(t)$, similar to \cite[Sec.~V-A]{che2021does}, and considering $\epsilon=10^{-5}$.

% \begin{figure}[!t]
%     \centering
%     \input{FigTikz/rolloff_maxER.tikz}
%     \caption{Maximum attainable $\ER$ vs the pulse roll-off factor. The admissible regions correspond to a feasible $\ER$ that satisfies the positivity constraint of the optical power \eqref{eq:P_tx}.}
%     \label{fig:rolloff_ER}
% \end{figure}

Next, we consider a laser in continuous-wave (CW) operation alongside a Mach-Zehnder modulator (MZM) to modulate the amplitude of the optical wave. The output optical power of the laser can be modeled as \cite[Eq.~(2.2)]{seimetz2009high}
\begin{equation} \label{eq:Pcw}
    P_{\mathrm{cw}}(t) = \Pcw\left(1 + N_{\rin}(t)\right),
\end{equation}
where $\Pcw$ is the laser's average optical power and $N_{\rin}(t)$ models the RIN. Following \cite[Ch.~3.3.6]{hui2019introduction}, we model RIN as a continuous-time zero-mean white Gaussian random process with \mbox{$\E{N_{\rin}(t)}\!=\!0$}, \mbox{$R_{{\rin}}(\tau)\!=\!\frac{1}{2}N_0^\rin\delta(\tau)$}, and \mbox{$S_\rin(f)\!=\!\frac{1}{2}N_0^\rin$}, where $R_\rin(\tau)$ and $S_\rin(f)$ are the autocorrelation function and double-sided PSD of $N_{\rin}(t)$, respectively. 

%\subsubsection{MZM Pre-Distortion}
At the fiber input, the transmit signal $P_{\tx}(t)$ should be as close as possible to the PAM waveform $S(t)$. The device used for the electrical-to-optical conversion is an MZM in push-pull configuration. The amplitude of the signal $S(t)$ is pre-distorted to account for the nonlinear MZM transfer function. The pre-distortion step (``Pre Dist.'' block in Fig.~\ref{fig:system_IMDD}) uses the inverse MZM transfer function and normalizes the amplitude by the laser power $\Pcw$, similar to \cite[Eq.~(9)]{napoli2017digital}. Without loss of generality, we consider that the output of the modulator is
\begin{align}
    P_\tx(t) = S(t)\left(1+N_\rin(t)\right). \label{eq:P_tx}
\end{align}
Although $N_\rin(t)$ has infinite support, RIN values for lasers typically range from $[-160,-130]$ dB/Hz \cite[Ch.~2.1.1]{seimetz2009high}. Hence, for a sampling rate of $2B$, $\mathrm{Pr}\{N_\rin(k/2B) < -1\}< 10^{-13}$, and thus, noise samples are assumed to not violate the positivity constraint of the optical power in \eqref{eq:Pcw} and \eqref{eq:P_tx}.\footnote{For a RIN value of $\frac{1}{2}N_0^\rin=10^{-14}$ ($-140$ dB/Hz) over a bandwidth of $B=10^{12}$ Hz, $\mathrm{Pr}\{N_\rin(k/2B)<-1\}=\mathsf{Q}\left(({N_0^\rin B})^{-1/2}\right)\approx 10^{-13}$, with $\mathsf{Q}(\cdot)$ being the Gaussian Q-function.} 

The main application of high-speed IM-DD systems is intra-DCI for very limited reach, where the fiber length is expected to be $L\!<\!1$ km for 400 Gbps and beyond. We assume that the laser is operated in the O-band with wavelength $1310$ nm, where chromatic dispersion is negligible \cite[Sec.~3.3]{wei2020experimental}. Thus, we only consider attenuation during fiber propagation, i.e., \mbox{$P_{\rx}(t) = \varphi P_\tx(t)$}, where $10\log_{10}(1/\varphi)=\alpha L$ where $\alpha=0.35$ dB/km is the attenuation coefficient at $1310$ nm.%\footnote{Another application where high-speed IM-DD transceivers are used is free-space optical (FSO) communications. Our system model can be adapted for FSO links by adjusting $\alpha$ considering atmospheric attenuation and turbulence.}

%\section{Channel Model}\label{sec:chanmodel}

At the receiver, a single photodiode (PD) is used to generate a current $I(t)=\mathfrak{R}P_{\rx}(t)=\mathfrak{R}\varphi P_\tx(t)$, where $\mathfrak{R}$ is the responsivity in A/W. The current is then amplified and converted to a voltage $V(t)$ with a transimpedance amplifier (TIA), with gain $G$ in units of $\Omega$. The gain is set to $G=(\varphi\mathfrak{R})^{-1}$ to compensate for losses and PD responsivity. Also, during this process, thermal noise is generated which is modeled by $N_\thn(t)$, a zero-mean white Gaussian random process with \mbox{$\E{N_{\thn}(t)}\!=\!0$}, and \mbox{$S_\thn(f)\!=\!\frac{1}{2}N_0^\thn$}.
% \begin{equation}
%     \E{N_{\thn}(t)}\!=\!0,\;\ R_{{\thn}}(\tau)\!=\!\frac{N_0^\thn}{2}\delta(\tau),\;\ S_\thn(f)\!=\!\frac{N_0^\thn}{2}.
% \end{equation}
The output of the TIA is the voltage
\begin{align}
    V(t) &= G\left(\mathfrak{R}P_\rx(t) + N_\thn(t)\right) \label{eq:V_t_1}\\
    %G\left(I(t) + N_\thn(t)\right)\\
    %&= 
    %&= G\left(\mathfrak{R}\varphi P_\tx(t)+N_\thn(t)\right)\\
    %&= P_\tx(t) + GN_\thn(t)\\
    &= S(t) + S(t)N_\rin(t) + GN_\thn(t), \label{eq:V_t_2}
\end{align}
where we used \eqref{eq:P_tx} to get from \eqref{eq:V_t_1} to \eqref{eq:V_t_2}. The voltage $V(t)$ is then passed through the receiver filter $h(t)$. The filter $h(t)$ has unit gain and \mbox{$\|h\|_2^2=\int_{\mathbb{R}}h^2(t)\d t=R_\mathrm{s}$}. After filtering, samples are taken at the symbol rate, generating the observations $Y_k$ (see Fig.~\ref{fig:system_IMDD}).

% \textcolor{red}{We assume that the equalization stage is capable of fully removing the ISI arising due to the bandwidth limitation of the electronics. Thus, the analyzed equivalent channel is ISI-free \cite[Sec.~IV]{che2023modulation}, and we focus solely on the effect of the RIN on the received signal.}
\begin{theorem}[\textit{Channel Model}] \label{theo:chanmod}
The discrete-time equivalent channel model for the system in Fig.~\ref{fig:system_IMDD} is given by
\begin{equation} \label{eq:Y_k}
    Y_k = X_k + Z_k + Q_k,
\end{equation}
where $Q_k$ are independent and identically distributed Gaussian samples \mbox{$Q_k\!\sim\!\mathcal{N}\left(0,\sigma_Q^2\right)$} with $\sigma_Q^2=\frac{1}{2}N_0^\thn G^2 \|h\|_2^2$, and $Z_k$ is a signal-dependent noise term given by
\begin{align} \label{eq:Zk} \hspace{-4pt}
    Z_k &=\! \sum_{n\in\mathbb{Z}} X_n\! \underbrace{\int_{-\infty}^{\infty}\!\!p((k-n)T-\tau)N_\rin(kT-\tau)h(\tau)\d\tau}_{=: W_{k,n}}.
\end{align}
\end{theorem}

\begin{proof}
The output of the receiver filter $Y(t)=(V\!\ast\! h)(t)$ is
\begin{align} \label{eq:y_t}
    Y(t)=(S\!\ast \!h)(t) + \underbrace{\left((S\!\cdot\! N_\rin)\! \ast\! h\right)\!(t)}_{=: Z(t)} + \underbrace{G(N_\thn\!\ast\! h)(t)}_{=: Q(t)},
\end{align}
where $Z(t)$ denotes the filtered signal-dependent noise component and $Q(t)$ the filtered thermal noise. 
Since $p(t)$ and $h(t)$ form a pair of filters that satisfy the zero ISI Nyquist criterion, after sampling at $t=kT$ ($k\in\mathbb{Z}$), the symbols $X_k$ are recovered from \mbox{$(S\ast h)(kT)$} in \eqref{eq:y_t}.
Defining $Z_k=Z(kT)$ and $Q_k=Q(kT)$ as the sampled versions of $Z(t)$ and $Q(t)$ in \eqref{eq:y_t}, respectively, we obtain \eqref{eq:Y_k}.
\end{proof}

From \eqref{eq:Zk}, we observe that $Z_k$ depends on the current and neighboring symbols $X_n$ and their corresponding time-shifted pulses \mbox{$p((k-n)T-\tau)$}. Thus, we see via \eqref{eq:Zk} that \eqref{eq:Y_k} is not a memoryless channel model, as is usually assumed in the literature, see, e.g., \cite{szczerba20124,rizzelli2023analytical,maneekut2024reach}. Furthermore, since $Z_k$ is the mixture of two independent random variables, it is not guaranteed that $Z_k$ is also Gaussian. Although \eqref{eq:Zk} can be written in the form $Z_k=\sum_{n\in\mathbb{Z}} X_n W_{k,n}$, we emphasize that this is not the typical ISI setup for band-limited channels, as the coefficients $W_{k,n}$ are random due to the term $N_\rin(kT-\tau)$ inside the integral of \eqref{eq:Zk}.

Although the channel in \eqref{eq:Y_k} generally has memory, the memory coefficients are stochastic. Thus, here we consider the case where we treat the current symbol $X_k$, and any neighboring symbol $X_n$ ($n\neq k$) interference in the signal-dependent noise component $Z_k$ as additive noise. This assumption will also be used for the analysis in Sec.~\ref{sec:air} and is explored in the following example.

\begin{figure}[t]
    \centering
    \pgfplotstableread{FigTikz/data_txt/histogram_fit_rho0.1.txt}\histogramBoth
\pgfplotstableread{FigTikz/data_txt/histogram_rin_rho0.1.txt}\histogramRIN
\pgfplotstableread{FigTikz/data_txt/histogram_thn_rho0.1.txt}\histogramTHN

\def\lw{0.4pt}
\def\opac{0.4}

% \definecolor{Set1}{RGB}{27,158,119}
% \definecolor{Set2}{RGB}{231,41,138}
% \definecolor{Set3}{RGB}{55,126,184}
\definecolor{Set3}{RGB}{27,158,119}
\definecolor{Set2}{RGB}{161,33,33}
\definecolor{Set1}{RGB}{69,99,168}

%\definecolor{Set3}{RGB}{27,158,119}

\begin{tikzpicture}
    \begin{axis}[
        width=1.025\columnwidth, %% make sure it fits 
        height=1.6in,
        font=\footnotesize,
        xlabel={Channel Output $Y_k$ \scriptsize{$(\cdot10^{-3})$}},
        ylabel={Probability Density \scriptsize{$(\cdot10^3)$}},
        xmin=-4, xmax=4,
        ymin=0, ymax=1.1*6,
        y filter/.code={\pgfmathparse{#1*6}\pgfmathresult},
        xtick={-4,-3,...,4},
        xticklabels={,$0.55$,,$0.88$,,$1.22$,,$1.55$,,},
        grid=major,
        grid style = {solid,lightgray!50},
        legend style = {at={(0.5,1.00)}, anchor=south, font=\footnotesize, legend cell align=left, row sep=-0.5ex, column sep=0.5ex,inner sep=0.2ex},
        legend columns=-1,
        ylabel style={yshift=-0.5ex},
        xlabel style={yshift=0.5ex},
        ]

        \addplot[name path=baseline, forget plot] coordinates{(-4,0) (4,0)};
        
        \addplot[Set1, const plot, line width=\lw, name path=histTHN, opacity=\opac, forget plot] table[x index=0, y index=1]{\histogramTHN};
        \addplot[Set1, opacity=\opac] fill between[of=histTHN and baseline];
        \addlegendentry{$Y_k$ (Only $Q_k)\;$};
        
        \addplot[Set2, const plot, line width=\lw, name path=histRIN, opacity=\opac, forget plot] table[x index=0, y index=1]{\histogramRIN};
        \addplot[Set2, opacity=\opac] fill between[of=histRIN and baseline];
        \addlegendentry{$Y_k$ (Only $Z_k)\;$};
        
        \addplot[Set3, const plot, line width=\lw, name path=histBoth, opacity=\opac, forget plot] table[x index=0, y index=1]{\histogramBoth};
        \addplot[Set3, opacity=\opac] fill between[of=histBoth and baseline];
        \addlegendentry{$Y_k$};       
        
    \end{axis}
\end{tikzpicture}
    \vspace{-0.5em}
    \caption{Histogram of the samples $Y_k$. The contributions of the thermal noise $Q_k$ and signal-dependent noise $Z_k$ are also plotted independently.}
    \label{fig:histograms}
    \vspace{-0.5em}
\end{figure}

\begin{example}[\textit{Signal-dependent Noise Simulations}]\label{ex:histo}
We consider the noise component $Z_k$ by performing a continuous-time simulation of the model in Fig.~\ref{fig:system_IMDD} with the equally-spaced constellation \mbox{$\Xset=10^{-3}\{0.55, 0.88, 1.22, 1.55\}$}, along with the parameters specified in Table~\ref{tab:sim_param}. Figure~\ref{fig:histograms} shows the histogram of the received samples $Y_k$. The effect of $Q_k$ is shown by \textit{turning off} $Z_k$, and vice versa. From the histograms, it is clear that the variance of $Z_k$ depends on the amplitude of the transmitted symbol $X_k=x$, i.e., the larger the value of $x$, the larger the variance. 
\end{example}

While it is straightforward to see that \mbox{$Q_k\!\sim\!\mathcal{N}\left(0,\sigma_Q^2\right)$}, the distribution of $Z_k$ is not trivial to determine due to its dependence on the transmitted signal $S(t)$. The statistical properties of $Z_k$ are analyzed next.

% \textcolor{red}{{\bf From RIN paper, possibly repetitive:}
% The discrete-time channel model is given by
% \begin{equation}
%     Y_k = A_k + Z_k + Q_k,
%     % \label{eq:yk}
% \end{equation}
% where \( A_k \in \{a_0, a_1, \ldots, a_{M-1}\} \subset \mathbb{R}^+ \), 
% \( Q_k \sim \mathcal{N}(0, \sigma_Q^2) \), and
% \begin{align}
%     Z_k &= \sum_{m=-\infty}^{\infty} A_m \, S_{k,m}, \\
%     S_{k,m} &= \int_{-\infty}^{\infty} p((k-m)T - \tau) \, p(-\tau) \, N_{\mathrm{RIN}}(kT - \tau) \, d\tau.
% \end{align}
% \begin{itemize}
%     \item \( N(t) \) is a white Gaussian process with PSD \( N_0^{\mathrm{RIN}}/2 \).
%     \item Note that \( S_{k,m} \) is a random variable because of the presence of \( N_{\mathrm{RIN}}(kT - \tau) \) in the integral.
%     \item If \( N_{\mathrm{RIN}}(kT - \tau) \) were not present, then
%     \[
%         S_{k,m} = \int_{-\infty}^{\infty} p((k-m)T - \tau)p(-\tau) \, d\tau =
%         \begin{cases}
%             1, & k = m,\\
%             0, & k\neq m,
%         \end{cases}
%     \]
%     implying a memoryless channel.
% \end{itemize}
% }
\section{Signal-dependent Noise Analysis}\label{sec:noisestat}

The signal-RIN interaction term is often modeled as \mbox{$Z_k\sim\mathcal{N}(0,\sigma^2_{Z}(X_k))$}, where% $\sigma_{Z}^2(\mathsf{a}_n)$ depends only on the square of the transmit symbol, i.e.,
%In the literature, the noise standard model is given by the assumption that $Z_k$ follows a Gaussian distribution \mbox{$Z_k\sim\mathcal{N}(0,\sigma^2_{Z_k})$}, where the variance $\sigma_{Z_k}^2$ is
\begin{equation}\label{eq:model_old0}
    \sigma_{Z}^2(x)=x^2\frac{N_0^\rin}{2}\|h\|_2^2,
\end{equation}
and $x$ is the value of the transmitted symbol at instant $X_k=x$. When $h(t)$ is an ideal brick-wall filter with bandwidth $B$ and unit gain, $\|h\|_2^2=2B$ and \eqref{eq:model_old0} results in $\sigma_{Z}^2(x)=x^2 N_0^\rin  B$, which is the expression most often used (e.g., in \cite[Eq.~(10)]{szczerba20124},\cite[Eq.~(8)]{rizzelli2023analytical}, \cite[Eq.~(8)]{maneekut2024reach}, and \cite[Eq.~(8.2.7)]{hui2019introduction}), i.e., $\sigma_{Z}^2(X_k)$ depends only on the square of the transmit symbol. If both $p(t)$ and $h(t)$ are rectangular pulses in time domain, we see from \eqref{eq:Zk} that $Z_k$ depends only on $X_k$, i.e., the channel is memoryless, and the variance \eqref{eq:model_old0} is correct. However, this implies that $p(t)$ and $h(t)$ have infinite bandwidth. The rest of the paper is therefore focused on practical band-limited pulses.

\begin{table}[t]
\renewcommand{\arraystretch}{1.1}
\centering
\caption{Simulation Parameters and Values}
\label{tab:sim_param}
\vspace{-0.5em}
\begin{threeparttable}
    \begin{tabular}{lclc}
    \toprule
    \textbf{Parameter} & \textbf{Value} & \textbf{Parameter} & \textbf{Value}\\
    \midrule
    Modulation order $M$ & $4$ & $\ER$ & $4.5$ dB\\
    Symbol rate $R_\mathrm{s}$ & $225$ GBd & Noise $N_0^\rin$ & $-140$ dB/Hz\\
    Filters $p(t)$ and $h(t)$ & RRC & Fiber length $L$ & $1$ km\\
    RRC roll-off & $0.1$ & Attenuation $\alpha$ & $0.35$ dB/km\\
    Samples per symbol & $4$ & Responsivity $\mathfrak{R}$ & $0.5$ A/W\\
    %$\OMA$ & $0$ dBm & Noise $(N_0^\thn)^{1/2}$ & $22$ pA/$\sqrt{\text{Hz}}$\\
    $\OMA$ & $0$ dBm & Noise $N_0^\thn$ & $-183$ dBm/Hz\\
    \bottomrule
    \end{tabular}
\end{threeparttable}
\vspace{-1em}
\end{table}

Motivated by the results in Example~\ref{ex:histo}, we are now interested in the conditional first and second moments of the noise samples $Z_k$, conditioned on a given transmitted symbol $X_k=x$. Using \eqref{eq:Zk}, it is straightforward to see that the conditional mean \mbox{$\E{Z_k\!\mid \!X_k\!=\!x}=0$} because $N_\rin(t)$ is zero-mean. The conditional variance \mbox{$\E{Z_k^2\!\mid \!X_k\!=\!x}$}, on the other hand, is given by the following theorem.

\begin{theorem}[\textit{Conditional Variance of $Z_k$}] \label{theo:variance}
    The variance of $Z_k$ conditioned on the symbol $X_k=x$ is
    \begin{align} \label{eq:Zk_var} %\hspace{-3pt}
        \E{Z_k^2|X_k\!=\!x} &= \frac{N_0^\rin}{2}\!\sum_{i\in\mathbb{Z}}\sum_{j\in\mathbb{Z}}\E{X_iX_j\!\mid \!X_k\!=\!x}\Psi_{i,j},
    \end{align}
    where $\Psi_{i,j}$ is defined as
    \begin{align} \label{eq:Psi}
        \Psi_{i,j} = \int_{-\infty}^{\infty}p((k-i)T-\tau)p((k-j)T-\tau)h^2(\tau)\d\tau,
    \end{align}
    and the conditional cross expectation in \eqref{eq:Zk_var} is given by
    \begin{equation} \label{eq:cross_E}
        \E{X_iX_j\!\mid\!X_k\!=\!x}=\begin{cases}
            m_1^2 & (i\neq j) \wedge (i\neq k \wedge j\neq k),\\
            m_1x & (i\neq j) \wedge (i = k \vee j = k),\\
            m_2 & i=j\neq k,\\
            x^2 & i=j=k,
        \end{cases}
    \end{equation}
    where $m_1=\E{X}>0$ and $m_2=\E{X^2}$.
\end{theorem}

\begin{proof}
    The expected value of \eqref{eq:Zk} squared is $\E{Z_k^2|X_k}=\sum_i\sum_j\E{X_iX_j|X_k}\E{W_{k,i}W_{k,j}}$. Given that \mbox{$R_{\rin}(\tau)=\frac{1}{2}N_0^\rin\delta(\tau)$}, the term $\E{W_{k,i}W_{k,j}}=\frac{1}{2}N_0^\rin \Psi_{i,j}$. For the first two cases in \eqref{eq:cross_E}, i.e., $i\!\neq\! j$, \mbox{$\E{X_iX_j|X_k}=\E{X_i|X_k}\E{X_j|X_k}$}. For the last two cases in \eqref{eq:cross_E}, i.e., $i\!=\!j$, \mbox{$\E{X_iX_j|X_k}=\E{X_i^2|X_k}$}. 
\end{proof}

Note that $\Psi_{i,j}$~\eqref{eq:Psi} is a deterministic quantity, depending only on $p(t)$ and $h(t)$. From \eqref{eq:cross_E}, it can be seen that the variance \eqref{eq:Zk_var} will have a polynomial form \mbox{$\E{Z_k^2\!\mid\!X_k\!=\!x}=\frac{1}{2}N_0^\rin(p_0+p_1x+p_2x^2)$}, where $p_i$ are coefficients that depend on $\Psi_{i,j}$, $m_1$, and $m_2$. Hence, we see from \eqref{eq:Zk_var} that the variance depends not only on $x^2$ as in the traditional model in \eqref{eq:model_old0}, but also on a constant $p_0$ and a scaled version of $x$. 

\begin{example}[\textit{Conditional Variance Calculation}]\label{ex:cond-variance}
%Figure~\ref{fig:s2rin_memory} shows \eqref{eq:Zk_var} as a function of the one-sided memory length $\ell$ for the same parameters in Table~\ref{tab:sim_param}. The true variance \mbox{$\E{Z_k^2|A_k\!=\!\mathsf{a}_n}$}, with $\ell\to\infty$ in \eqref{eq:Zk_var} (black solid curves), is calculated directly from \eqref{eq:Y_k} using a \textit{genie-aided} approach. The colored curves in Fig.~\ref{fig:s2rin_memory} show that \eqref{eq:Zk_var} converges to the true value when considering $\ell\geq 12$. Furthermore, using \eqref{eq:model_old0} (black dashed lines) results in an incorrect variance estimate.
Figure~\ref{fig:s2rin_memory} shows the comparison of calculating the variance using either the common model of \eqref{eq:model_old0} or the conditional variance of \eqref{eq:Zk_var} for a Gaussian fit of the channel samples $Y_k$ obtained from the simulation of Fig.~\ref{fig:system_IMDD}. From the figure, it is clear that \eqref{eq:Zk_var} (green curve) accurately characterizes the noise variance from the channel, while \eqref{eq:model_old0} (red curve) underestimates the variance for the first and second symbols, and overestimates it for the third and fourth ones.
\end{example}

\begin{figure}[t]
    \centering
    \pgfplotstableread{FigTikz/data_txt/histogram_match.txt}\datafit

\def\lw{0.7pt}
\def\opac{0.4}

% \definecolor{Set1}{RGB}{27,158,119}
% \definecolor{Set2}{RGB}{231,41,138}
% \definecolor{Set3}{RGB}{55,126,184}
\definecolor{Set3}{RGB}{27,158,119}
\definecolor{Set2}{RGB}{161,33,33}
\definecolor{Set1}{RGB}{69,99,168}

%\definecolor{Set3}{RGB}{27,158,119}

\begin{tikzpicture}
    \begin{axis}[
        width=1.025\columnwidth, %% make sure it fits 
        height=1.6in,
        font=\footnotesize,
        xlabel={Channel output $Y_k$ \scriptsize{$(\cdot10^{-3})$}},
        ylabel={Probability Density \scriptsize{$(\cdot10^3)$}},
        y filter/.code={\pgfmathparse{#1*1e-3}\pgfmathresult},
        x filter/.code={\pgfmathparse{#1*1e3}\pgfmathresult},
        xmin=0.4, xmax=1.8,
        ymin=-0.04, ymax=4.5,
        ytick={0,1,...,4},
        grid=major,
        grid style = {solid,lightgray!50},
        legend style = {at={(0.99, 0.98)}, anchor=north east, font=\footnotesize, legend cell align=left, row sep=-0.5ex, column sep=0.5ex,inner sep=0.4ex},
        legend columns=-1,
        ylabel style={yshift=-0.5ex},
        xlabel style={yshift=0.5ex},
        ]

        \addplot[name path=baseline, draw=none, forget plot] coordinates{(0,-1) (2e-3,-1)};
        
        \addplot[Set1, draw=none, line width=\lw, name path=hist, opacity=\opac, forget plot] table[x index=0, y index=1]{\datafit};
        \addplot[Set1, opacity=\opac] fill between[of=hist and baseline];
        \addlegendentry{Simulation};

        \addplot[Set2, line width=\lw] table[x index=2, y index=3]{\datafit};
        \addlegendentry{Eq. \eqref{eq:model_old0}};

        \addplot[Set3, line width=\lw] table[x index=2, y index=4]{\datafit};
        \addlegendentry{Eq. \eqref{eq:Zk_var}};
        
        % \addplot[Set2, const plot, line width=\lw, name path=histRIN, opacity=\opac, forget plot] table[x index=0, y index=1]{\histogramRIN};
        % \addplot[Set2, opacity=\opac] fill between[of=histRIN and baseline];
        % \addlegendentry{$Y_k$ (Only $Z_k)\;$};
        
        % \addplot[Set3, const plot, line width=\lw, name path=histBoth, opacity=\opac, forget plot] table[x index=0, y index=1]{\histogramBoth};
        % \addplot[Set3, opacity=\opac] fill between[of=histBoth and baseline];
        % \addlegendentry{$Y_k$};       
        
    \end{axis}
\end{tikzpicture}
    \vspace{-0.5em}
    \caption{Comparison of the variance calculation using the common model $\sigma_Z^2(x)$ of \eqref{eq:model_old0} and the developed expression $\E{Z_k^2|X_k\!=\!x}$ of \eqref{eq:Zk_var}.}
    \label{fig:s2rin_memory}
    \vspace{-0.5em}
\end{figure}

% \begin{figure}[t]
% \vspace{0.1cm}
%     \centering
%     \input{FigTikz/s2rin_memory.tikz}
%     \caption{Conditional variance $\E{Z_k^2|A_k\!=\!\mathsf{a}_n}$ using \eqref{eq:Zk_var} as function of $\ell$.}
%     \label{fig:s2rin_memory}
% \end{figure}
\section{Information-Theoretic Analysis}\label{sec:air}% of the RIN-Dominated IM/DD Channel}

The channel in \eqref{eq:Y_k} is a channel with memory, and thus, the AIR of interest is the mutual information rate of the input and output defined as~\cite[Eq. (2)]{soriaga2007determining}
\begin{equation}
    \mathcal{I}(X;Y) \triangleq 
    \lim_{N \to \infty} \frac{1}{N} I(\boldsymbol{X}; \boldsymbol{Y}),
    \label{eq:mir}
\end{equation}
where $I(\cdot ; \cdot)$ is the mutual information (MI), $\boldsymbol{X} = (X_1, X_2, \ldots, X_N)$, and $\boldsymbol{Y} = (Y_1, Y_2, \ldots, Y_N)$. The MI is an AIR using an optimal decoding metric proportional to $p_{\boldsymbol{Y}|\boldsymbol{X}}(\boldsymbol{y}|\boldsymbol{x})$. The channel capacity is then the supremum of $\mathcal{I}(X;Y)$ over all possible input distributions $P_{\boldsymbol{X}}(\boldsymbol{x})$, potentially with memory.
% \begin{align}
%     \hat{\boldsymbol{x}} 
%     &= \argmax_{\boldsymbol{x} \in \mathcal{X}^N} p_{\boldsymbol{X},\boldsymbol{Y}}(\boldsymbol{x}, \boldsymbol{y})= \argmax_{\boldsymbol{x} \in \mathcal{X}^N} p_{\boldsymbol{Y}|\boldsymbol{X}}(\boldsymbol{y}|\boldsymbol{x}),
% \end{align}
% where we used the fact that $P_{\boldsymbol{X}}(\boldsymbol{x})$ is constant.
% The optimal decoding metric is therefore proportional to $p_{\boldsymbol{Y}|\boldsymbol{X}}(\boldsymbol{y}|\boldsymbol{x})$.

\subsection{Lower Bound via Generalized Mutual Information}

A lower bound on $ \mathcal{I}(X;Y) $ can be obtained by using a mismatched decoding metric $ \qmetric^N(\boldsymbol{x}, \boldsymbol{y}) $ defined as
\begin{equation}
    %\tilde{\qmetric}(\boldsymbol{x}, \boldsymbol{y}) = \prod_{n=1}^{N} \qmetric(x_n, y_n),\label{eq:qay}
    \qmetric^N(\boldsymbol{x}, \boldsymbol{y}) = \prod_{n=1}^{N} \qmetric(x_n,y_n).\label{eq:qay}
\end{equation}%where $\qmetric(x, y) = 1/M \cdot \qmetric(y|x)$.
 Decoding using the product metric in~\eqref{eq:qay} is suboptimal since it is not proportional to $p_{\boldsymbol{Y}|\boldsymbol{X}}(\boldsymbol{y}|\boldsymbol{x})$. This approach leads to the generalized mutual information (GMI) $I_{\qmetric}^{\mathrm{gmi}}$, which lower bounds the mutual information rate in \eqref{eq:mir} as~\cite{Kaplan1993_RatesExponentsCompondChan,Merhav1994_InfoRatesMismDecod,Martinez2009_TransIT_bicm} 
\begin{equation}\label{eq:lowerbound}
    \mathcal{I}(X;Y) 
    \geq I_{\qmetric}^{\mathrm{gmi}}(X;Y)
    \triangleq \max_{s\geq0} I_{\qmetric,s}^{\mathrm{gmi}}(X;Y),
\end{equation}
where~\cite[Eq. (4.35)]{Szczecinski2015_BICMbook},~\cite[Eq.~(11)]{scarlett2014mismatched}%~\cite[Eq. (17)]{alvarado2018applicationsandcomputations}
% \begin{align}
%     I_{\qmetric,s}^{\mathrm{gmi}}(X;Y)
%     &\triangleq \E{
%         \log_2
%         \frac{[\qmetric(X,Y)]^s}{\sum_{x'\in\Xset} P_X(x') [\qmetric(x',Y)]^s}
%     } \notag\\
%     &= \frac{1}{M} \sum_{i=0}^{M-1} 
%        \int_{-\infty}^{\infty}\!\! p_{Y|X}(y|x_i)
%        \log_2 \frac{M[\qmetric(y|x_i)]^s}
%        {\sum_{j=1}^{M} [\qmetric(y|x_j)]^s} \, \mathrm{d}y \notag\\
%     &= \log_2 M + \sum_{i=0}^{M-1}\beta_i,
%     \label{eq:H}
% \end{align}
\begin{align}
    I_{\qmetric,s}^{\mathrm{gmi}}(X;Y)
    &\triangleq \E{
        \log_2
        \frac{[\qmetric(X,Y)]^s}{\frac{1}{M}\sum_{x'\in\Xset} [\qmetric(x',Y)]^s}
    } \notag\\
    &= \log_2 M + \sum_{i=1}^{M}\beta_i,
    \label{eq:H}
\end{align}
and where
\begin{align}\label{beta}
\beta_i \triangleq \frac{1}{M}
       \int_{-\infty}^{\infty} p_{Y|X}(y|x_i)
       \log_2 \frac{[\qmetric(x_i,y)]^s}
       {\sum_{j=1}^{M} [\qmetric(x_j,y)]^s} \, \mathrm{d}y.
\end{align}
For the channel in Theorem~\ref{theo:chanmod}, we choose the metric
\begin{equation}
    \qmetric(x,y) = \frac{1}{\sqrt{\sigma_Q^2 + \sigma_Z^2(x)}}
        \exp \!\left(-\frac{(y-x)^2}{2(\sigma_Q^2 + \sigma_Z^2(x))}\right), \label{eq:py_a}
\end{equation}
with $\sigma_Z^2(x) = \E{Z_k^2 \mid X_k = x }$ given by Theorem~\ref{theo:variance}. The metric in~\eqref{eq:py_a} therefore adds a second level of mismatch, on top the memoryless mismatch in \eqref{eq:qay}. The metric in~\eqref{eq:py_a} makes the assumption that when conditioning on $x$, the channel noise can be modeled as Gaussian, which might not be true. Nevertheless, \eqref{eq:lowerbound}--\eqref{eq:py_a} give a lower bound on the channel capacity, which is what we show in the next section.

\subsection{Numerical Results}

\begin{figure}[t]
    \centering
    \resizebox{\columnwidth}{!}{\begin{tikzpicture}

\definecolor{set1}{RGB}{161,33,33}
\definecolor{set2}{RGB}{69,99,168}
\definecolor{set3}{RGB}{27,158,119}
\definecolor{set4}{RGB}{152,78,163}
\definecolor{set5}{RGB}{255,127,0}
\definecolor{set6}{RGB}{166,118,29}
\definecolor{set7}{RGB}{230,171,2}
\definecolor{set8}{RGB}{231,41,138}

\def\lw{0.75pt}

\begin{axis}[
every axis/.append style={font=\footnotesize},
xlabel={OMA (dBm)},
ylabel={GMI $I_\qmetric^{\mathrm{gmi}}$ [bit/sym]},
width=\columnwidth,  % Same width
height=0.70\columnwidth, % Same height
ylabel near ticks,
xlabel near ticks,
xtick={-25,-20,...,25},
ytick={0,0.5,...,3},
xmin=-25,
xmax=10,
grid style={solid,lightgray!75},
ymin=0,
ymax=3,
grid=major,
legend style={font=\scriptsize, legend cell align=left, at={(0.01,0.99)}, anchor=north west,row sep=-0.5ex},
]

\addlegendimage{black, only marks, line width=\lw, mark=asterisk, mark options={solid,scale=1, fill=white, line width=0.6pt}};
\addlegendentry{$\sigma_Z^2(x)$ in \eqref{eq:model_old0}};
\addlegendimage{black, no marks, line width=2.2*\lw, mark=asterisk, mark options={solid,scale=0.8, fill=white, line width=0.6pt}};
\addlegendentry{$\sigma_Z^2(x)$ in Thm.~\ref{theo:variance}};

\addplot [set5, line width=\lw, mark=none, mark options={solid,scale=0.4, fill=white, line width=0.6pt}]  file[]{data/GMI_vs_OMA_32-PAM_flag_Q_1_flag_Z_1.txt};\addlegendentry{$M = 32$};

% > 16-PAM
\addplot [set4, line width=\lw, mark=none, mark options={solid,scale=0.4, fill=white, line width=0.6pt}]  file[]{data/GMI_vs_OMA_16-PAM_flag_Q_1_flag_Z_1.txt};\addlegendentry{$M = 16$};

% > 8-PAM
\addplot [set3, line width=\lw, mark=none, mark options={solid,scale=0.4, fill=white, line width=0.6pt}]  file[]{data/GMI_vs_OMA_8-PAM_flag_Q_1_flag_Z_1.txt};\addlegendentry{$M = 8$};

% > 4-PAM
\addplot [set2, line width=\lw, mark=none, mark options={solid,scale=0.4, fill=white, line width=0.6pt}]  file[]{data/GMI_vs_OMA_4-PAM_flag_Q_1_flag_Z_1.txt};\addlegendentry{$M = 4$};

% > 2-PAM
\addplot [set1, line width=\lw, mark=none, mark options={solid,scale=0.4, fill=white, line width=0.6pt}]  file[]{data/GMI_vs_OMA_2-PAM_flag_Q_1_flag_Z_1.txt};\addlegendentry{$M = 2$};

% > 32-PAM
% \addplot [set5, line width=\lw, mark=none, mark options={solid,scale=0.4, fill=white, line width=0.6pt}]  file[]{data/GMI_vs_OMA_32-PAM_flag_Q_1_flag_Z_1.txt};\addlegendentry{$M = 32$};
\addplot [set5, only marks, line width=\lw, mark=asterisk, mark repeat=2, mark options={solid,scale=0.5, fill=white, line width=0.6pt}]  file[]{data/second_model/GMI_vs_OMA_32-PAM_flag_Q_1_flag_Z_1_2ndModel.txt};
% \addlegendentry{$M = 32$};

% > 16-PAM
% \addplot [set4, line width=\lw, mark=none, mark options={solid,scale=0.4, fill=white, line width=0.6pt}]  file[]{data/GMI_vs_OMA_16-PAM_flag_Q_1_flag_Z_1.txt};\addlegendentry{$M = 16$};
\addplot [set4, only marks, line width=\lw, mark=asterisk, mark repeat=2, mark options={solid,scale=0.5, fill=white, line width=0.6pt}]  file[]{data/second_model/GMI_vs_OMA_16-PAM_flag_Q_1_flag_Z_1_2ndModel.txt};
% \addlegendentry{$M = 16$};

% > 8-PAM
% \addplot [set3, line width=\lw, mark=none, mark options={solid,scale=0.4, fill=white, line width=0.6pt}]  file[]{data/GMI_vs_OMA_8-PAM_flag_Q_1_flag_Z_1.txt};\addlegendentry{$M = 8$};
\addplot [set3, only marks, line width=\lw, mark=asterisk, mark repeat=2, mark options={solid,scale=0.5, fill=white, line width=0.6pt}]  file[]{data/second_model/GMI_vs_OMA_8-PAM_flag_Q_1_flag_Z_1_2ndModel.txt};
% \addlegendentry{$M = 8$};

% > 4-PAM
% \addplot [set2, line width=\lw, mark=none, mark options={solid,scale=0.4, fill=white, line width=0.6pt}]  file[]{data/GMI_vs_OMA_4-PAM_flag_Q_1_flag_Z_1.txt};\addlegendentry{$M = 4$};
\addplot [set2, only marks, line width=\lw, mark=asterisk, mark repeat=2, mark options={solid,scale=0.5, fill=white, line width=0.6pt}]  file[]{data/second_model/GMI_vs_OMA_4-PAM_flag_Q_1_flag_Z_1_2ndModel.txt};
% \addlegendentry{$M = 4$};

% > 2-PAM
% \addplot [set1, line width=\lw, mark=none, mark options={solid,scale=0.4, fill=white, line width=0.6pt}]  file[]{data/GMI_vs_OMA_2-PAM_flag_Q_1_flag_Z_1.txt};\addlegendentry{$M = 2$};
\addplot [set1, only marks, line width=\lw, mark=asterisk, mark repeat=2, mark options={solid,scale=0.5, fill=white, line width=0.6pt}]  file[]{data/second_model/GMI_vs_OMA_2-PAM_flag_Q_1_flag_Z_1_2ndModel.txt};
% \addlegendentry{$M = 2$};

\end{axis}

\end{tikzpicture}}
    \caption{GMI vs. OMA for the system in Fig.~\ref{fig:system_IMDD} using the parameters in Table~\ref{tab:sim_param}.}
    \label{fig:GMI}
    \vspace{-0.5em}
\end{figure}

Figure~\ref{fig:GMI} shows GMI $I_\qmetric^{\mathrm{gmi}}$ for the parameters given in \mbox{Table~\ref{tab:sim_param}}.
We evaluated \eqref{eq:lowerbound}--\eqref{eq:py_a} using Monte-Carlo integration. %{p.AWGN}
Equally-spaced constellations $\Xset$ for different values of $M$ are derived using the constraints given in \eqref{eq:OMA}. The following observations are made from Fig.~\ref{fig:GMI}. Firstly, lower-cardinality constellations provide higher AIRs except for high OMA values, for which GMI saturates.
Secondly, GMI saturates at $m=\log_2(M)$ for $M\leq4$ only, while it saturates at smaller values for higher $M$.
This saturation for denser constellations can be explained by the signal-dependent noise with higher variance $\sigma^2_{Z}(x)$ in~\eqref{eq:Zk_var} (see Fig.~\ref{fig:histograms}) that high-power constellation points experience. Figure~\ref{fig:GMI} shows that when the channel is RIN-dominated, using $M$-PAM with $M>8$ provides only negligible gains when using a memoryless decoding metric, which is in agreement with recent experimental observations in the literature~\cite[Table I]{pang2020}.
Finally, Fig.~\ref{fig:GMI} also shows GMI results when the variance $\sigma_Z^2(x)$ is calculated via \eqref{eq:model_old0} instead (stars). These results show little difference to those based on Thm.~\ref{theo:variance}, which make us conjecture that taking the memory from $Z_k$ into account is key to achieve higher rates for this channel model.
        
Next, we study in detail the convergence of the GMI as OMA grows. To this end, we consider $32$-PAM with $s=1$ and analyze the values of $\beta_i$ in~\eqref{beta}.\footnote{The optimum value of $s$ in the maximization \eqref{eq:lowerbound} was found to be $s=1$ for medium and high OMA values.} We first show results for the case where the RIN noise is set to zero, i.e., we set $\sigma_Z^2(x)=0$. For this case, Fig.~\ref{fig:betai} shows results for OMA values of $0$, $2$, and $4$~dBm (gray circles). These results show that as OMA grows, the contribution of the symbols (via $\beta_i$) is symmetric with respect to the middle of the constellation. Furthermore, the outermost constellation points contribute more to the GMI. This is expected, as in this case we are effectively considering an AWGN channel, where the contributions of the inner- and outer-most points at high SNRs (equivalently, high OMAs) only depend on the number of nearest neighbors \cite{Alvarado12b}. Lastly, these results show a clear convergence of all values of $\beta_i$ to zero, which results in a GMI saturating at $\log_2{M}$. 

\begin{figure}[t]
    \centering
    \resizebox{\columnwidth}{!}{\begin{tikzpicture}

\definecolor{set1}{RGB}{161,33,33}
\definecolor{set2}{RGB}{69,99,168}
\definecolor{set3}{RGB}{27,158,119}
\definecolor{set4}{RGB}{152,78,163}
\definecolor{set5}{RGB}{255,127,0}
\definecolor{set6}{RGB}{166,118,29}
\definecolor{set7}{RGB}{230,171,2}
\definecolor{set8}{RGB}{231,41,138}

\def\lw{0.75pt}

\begin{axis}[
  yticklabel style={
    /pgf/number format/fixed,
    /pgf/number format/precision=3
  },
every axis/.append style={font=\footnotesize},
xlabel={Constellation Point Index $i$},
ylabel={$\beta_i$ in \eqref{beta}},
width=\columnwidth,  % Same width
height=0.7\columnwidth, % Same height
ylabel near ticks,
xlabel near ticks,
xmin=1,
xmax=32,
xtick={4,8,...,32},
grid style={solid,lightgray!75},
ymin=-0.125,
ymax=0,
grid=major,
legend columns=2,
legend style={font=\scriptsize, legend cell align=left, at={(0.01,0.01)}, anchor=south west,row sep=-0.5ex}, 
]

\addplot [color=black!50!white,, line width=\lw, mark=*, mark options={solid,scale=0.8, fill=white, line width=1pt}]  file[]{FigTikz/data_txt/32-PAM_no_RIN_OMA_0dBm};\addlegendentry{OMA $0$~dBm (no RIN)};
\addplot [color=set1!50!white,, line width=\lw, mark=square*, mark options={solid,scale=0.7, fill=white, line width=1pt}]  file[]{FigTikz/data_txt/32-PAM_w_RIN_OMA_-5dBm};\addlegendentry{OMA $-5$~dBm};
\addplot [color=black!75!white, line width=\lw, mark=*, mark options={solid,scale=0.8, fill=white, line width=1pt}]  file[]{FigTikz/data_txt/32-PAM_no_RIN_OMA_2dBm};\addlegendentry{OMA $2$~dBm (no RIN)};
\addplot [color=set1!75!white,, line width=\lw, mark=square*, mark options={solid,scale=0.7, fill=white, line width=1pt}]  file[]{FigTikz/data_txt/32-PAM_w_RIN_OMA_0dBm};\addlegendentry{OMA $0$~dBm};
\addplot [color=black!1000!white,, line width=\lw, mark=*, mark options={solid,scale=0.8, fill=white, line width=1pt}]  file[]{FigTikz/data_txt/32-PAM_no_RIN_OMA_4dBm};\addlegendentry{OMA $4$~dBm (no RIN)};
\addplot [color=set1!100!white,, line width=\lw, mark=square*, mark options={solid,scale=0.7, fill=white, line width=1pt}]  file[]{FigTikz/data_txt/32-PAM_w_RIN_OMA_5dBm};\addlegendentry{OMA $5$~dBm};
%\addplot [color=red!80!white,, line width=\lw, mark=square*, mark options={solid,scale=0.7, fill=white, line width=1pt}]  file[]{FigTikz/data_txt/32-PAM_w_RIN_OMA_10dBm};\addlegendentry{OMA $10$~dBm};

\end{axis}

\end{tikzpicture}}
    \caption{Contribution of the constellation points via $\beta_i$ in~\eqref{beta} to the GMI for $M=32$.}
    \label{fig:betai}
\end{figure}

The second set of results in Fig.~\ref{fig:betai} (red squares) shows the case where RIN is present in the system. The contribution of the constellation points to the GMI in this case changes considerably. First of all, the symmetry with respect to the middle of the constellation is lost, which is caused by the signal-dependent noise. Furthermore, these results show that as OMA grows, the contributions of $\beta_i$ do not vanish, which explains why the GMI curves in Fig.~\ref{fig:GMI} saturate at values below $\log_2{M}$. This detailed convergence analysis could be used, e.g., to design optimal geometrical and/or probabilistic shaping signaling schemes.

To conclude, we analyze the saturation of the GMI for dense constellations and different RIN values. Figure~\ref{fig:convergence} shows these saturation values at very large OMA ($\text{OMA}=25$~dBm). The black line shows the GMI for the AWGN channel, which is equivalent to our setup with $\mathrm{RIN}\!\to\!-\infty$ dB/Hz. For finite RIN values, the figure shows that the GMI saturates as $M\to\infty$. As an example, for $-140$ dB/Hz (the GMI results shown in Fig.~\ref{fig:GMI}), the saturated GMI has a maximum at $M=16$, and confirms that using values of $M$ larger than $M=8$ provides little or no gain in terms of GMI. As mentioned above, we believe that using a mismatched metric with memory would remove this saturation.
%The results in this figure also show that , however, the gains over $M=8$ are marginal.
\begin{figure}[t]
    \centering
    %\pgfplotstableread{FigTikz/data_txt/GMI_vs_M.txt}\dataM
\pgfplotstableread{FigTikz/data_txt/GMI_vs_M.txt}\dataM

\begin{tikzpicture}

\definecolor{set1}{RGB}{161,33,33}
\definecolor{set2}{RGB}{69,99,168}
\definecolor{set3}{RGB}{27,158,119}
\definecolor{set4}{RGB}{152,78,163}
\definecolor{set5}{RGB}{255,127,0}
\definecolor{set6}{RGB}{166,118,29}
\definecolor{set7}{RGB}{230,171,2}
\definecolor{set8}{RGB}{231,41,138}

\def\lw{0.75pt}

\begin{axis}[
every axis/.append style={font=\footnotesize},
xlabel={Constellation Cardinality $M$},
ylabel={GMI $I_{\qmetric}^{\mathrm{gmi}}$ [bit/sym]},
width=\columnwidth,  % Same width
height=0.6\columnwidth, % Same height
ylabel near ticks,
xlabel near ticks,
xtick={1,2,...,10},
xticklabels={$2$,$4$,$8$,$16$,$32$,$64$,$128$,$256$,$512$,$1024$,$2048$,$4096$},
xmin=1,
xmax=10,
grid style={solid,lightgray!75},
xlabel style={yshift=1.0ex},
ymin=1.0,
ymax=7.0,
grid=major,
legend style={font=\scriptsize, legend cell align=left, at={(0.01,0.98)}, anchor=north west,row sep=-0.5ex},
%xmode=log,
%log ticks with fixed point,
xticklabel style={rotate=45, anchor=east}
]

%\addplot [black, dashed, line width=0.9*\lw, mark=none, forget plot]  coordinates{(2,2.8969) (256,2.8969)};

% \addplot [set1, dashed, line width=1.2*\lw, mark=none, forget plot]  coordinates{(2,2.83166) (256,2.83166)};

% ----- GMI Only Z_k -----
% > 2-PAM
% \addplot [black, line width=1.2*\lw, mark=*, mark options={solid,scale=0.7, fill=white, line width=0.6pt}]  coordinates{(2,1)
%     (4, 1.999)
%     (8, 2.8189)
%     (16, 2.9633)
%     (32, 2.9323)
%     (64, 2.9130)
%     (128, 2.9028)
%     (256, 2.8969)};
    % (512, 2.8936)
    % (1024, 2.8910)
    % (2048, 2.8915)
    % (4096, 2.8912)};

\addplot [black, line width=1.2*\lw, mark=none, mark options={solid,scale=0.7, fill=white, line width=0.6pt}]  coordinates{(1,1) (10, 10)};

\addplot [set4, line width=1.2*\lw, mark=*, mark options={solid,scale=0.7, fill=white, line width=0.6pt}] table[x index=0, y index=4]{\dataM};
\addplot [set3, line width=1.2*\lw, mark=*, mark options={solid,scale=0.7, fill=white, line width=0.6pt}] table[x index=0, y index=3]{\dataM};
\addplot [set2, line width=1.2*\lw, mark=*, mark options={solid,scale=0.7, fill=white, line width=0.6pt}] table[x index=0, y index=2]{\dataM};
\addplot [set1, line width=1.2*\lw, mark=*, mark options={solid,scale=0.7, fill=white, line width=0.6pt}] table[x index=0, y index=1]{\dataM};

% \addplot[set4, mark=*, mark options={scale=0.7}] coordinates{(7,5.2984)};
% \addplot[set3, mark=*, mark options={scale=0.7}] coordinates{(6,4.4918)};
% \addplot[set2, mark=*, mark options={scale=0.7}] coordinates{(5,3.7000)};
% \addplot[set1, mark=*, mark options={scale=0.7}] coordinates{(4,2.9436)};

\addplot[set4, mark=*, mark options={scale=0.7}] coordinates{(7,5.3106)};
\addplot[set3, mark=*, mark options={scale=0.7}] coordinates{(6,4.5032)};
\addplot[set2, mark=*, mark options={scale=0.7}] coordinates{(5,3.7120)};
\addplot[set1, mark=*, mark options={scale=0.7}] coordinates{(4,2.9544)};

\addlegendentry{AWGN};
\addlegendentry{$\mathrm{RIN}\!=\!-155$};
\addlegendentry{$\mathrm{RIN}\!=\!-150$};
\addlegendentry{$\mathrm{RIN}\!=\!-145$};
\addlegendentry{$\mathrm{RIN}\!=\!-140$};

%1.0000    1.9997    2.8145    2.9996    2.9729    2.9521    2.9412    2.9355    2.9334    2.9312
%\addplot[set1, densely dashed] coordinates{(1,1) (2,2) (3,2.8145) (4,2.9996) (5,2.9729) (6,2.9521) (7,2.9412) (8,2.9355) (9,2.9334) (10,2.9312)};

\end{axis}

\end{tikzpicture}
    \caption{GMI at $\text{OMA}=25$~dBm vs. $M$ for different RIN values (in dB/Hz). Filled markers denote the maximum GMI.}
    \label{fig:convergence}
    \vspace{-0.5em}
\end{figure}
\section{Conclusions}
\label{sec:conclusions}

We provided a novel discrete-time channel model with memory for optical communications with laser relative intensity noise. 
We derived an analytical expression for the variance of the signal-dependent noise that has a second-order polynomial form in terms of the transmit symbol. This model differs from the standard model in the literature where the variance depends only on the squared symbol. 
We showed that our model provides a more accurate representation of the continuous-time system.
We then studied the GMI for this channel with a memoryless decoding metric.
We observed that, for practical system parameters, using PAM modulation with more than $8$ points brings negligible gains. Future research avenues include constellation shaping and demapping with memory to improve AIRs of this channel.

\newpage
%\nocite{*}
\balance
\bibliographystyle{IEEEtran}
%\bibliography{IEEEabrv,references}
\bibliography{references}

\begin{thebibliography}{10}
\providecommand{\url}[1]{#1}
\csname url@samestyle\endcsname
\providecommand{\newblock}{\relax}
\providecommand{\bibinfo}[2]{#2}
\providecommand{\BIBentrySTDinterwordspacing}{\spaceskip=0pt\relax}
\providecommand{\BIBentryALTinterwordstretchfactor}{4}
\providecommand{\BIBentryALTinterwordspacing}{\spaceskip=\fontdimen2\font plus
\BIBentryALTinterwordstretchfactor\fontdimen3\font minus \fontdimen4\font\relax}
\providecommand{\BIBforeignlanguage}[2]{{%
\expandafter\ifx\csname l@#1\endcsname\relax
\typeout{** WARNING: IEEEtran.bst: No hyphenation pattern has been}%
\typeout{** loaded for the language `#1'. Using the pattern for}%
\typeout{** the default language instead.}%
\else
\language=\csname l@#1\endcsname
\fi
#2}}
\providecommand{\BIBdecl}{\relax}
\BIBdecl

\bibitem{zhou2019beyond}
X.~Zhou, R.~Urata, and H.~Liu, ``Beyond 1 {Tb/s} intra-data center interconnect technology: {IM-DD} or coherent?'' \emph{Journal of Lightwave Technology}, vol.~38, no.~2, pp. 475--484, 2020 2020.

\bibitem{che2023modulation}
D.~Che and X.~Chen, ``Modulation format and digital signal processing for {IM-DD} optics at post-200{G} era,'' \emph{Journal of Lightwave Technology}, vol.~42, no.~2, pp. 588--605, Jan. 2024.

\bibitem{wei2020experimental}
J.~Wei, T.~Rahman, S.~Calabr{\`o}, N.~Stojanovic, L.~Zhang, C.~Xie, Z.~Ye, and M.~Kuschnerov, ``Experimental demonstration of advanced modulation formats for data center networks on 200 {Gb/s} lane rate {IMDD} links,'' \emph{Optics Express}, vol.~28, no.~23, pp. 35\,240--35\,250, Nov. 2020.

\bibitem{baseline_200G}
B.~Welch, J.~Ingham, E.~Bernier, and P.~Dawe, ``Baseline proposals for {200G/L} {PMD} specifications for single wavelength 500 m and 2 km standards,'' \emph{IEEE P802.3dj Ethernet Task Force}, Feb. 2023.

\bibitem{agrawal2013semiconductor}
G.~P. Agrawal and N.~K. Dutta, \emph{Semiconductor lasers}.\hskip 1em plus 0.5em minus 0.4em\relax Springer Science \& Business Media, 2013.

\bibitem{RIN_Keysight}
{Keysight Technologies}, \emph{Digital Communication Analyzer (DCA), Measure Relative Intensity Noise (RIN)}, July 2014.

\bibitem{szczerba20124}
K.~Szczerba, P.~Westbergh, J.~Karout, J.~S. Gustavsson, {\AA}.~Haglund, M.~Karlsson, P.~A. Andrekson, E.~Agrell, and A.~Larsson, ``4-{PAM} for high-speed short-range optical communications,'' \emph{Journal of Optical Communications and Networking}, vol.~4, no.~11, pp. 885--894, Oct. 2012.

\bibitem{rizzelli2023analytical}
G.~Rizzelli, P.~Torres-Ferrera, F.~Forghieri, and R.~Gaudino, ``An analytical model for performance estimation in modern high-capacity {IMDD} systems,'' \emph{Journal of Lightwave Technology}, vol.~42, no.~5, pp. 1443--1452, Mar. 2023.

\bibitem{maneekut2024reach}
R.~Maneekut, D.~J. Elson, Y.~Wakayama, N.~Yoshikane, and P.~Kaewplung, ``Reach extension of a 53.12 {Gbps} {PAM}-4 {IM/DD} signals with optimized extinction ratio in amplified multi-span {O-Band} transmission,'' \emph{IEEE Access}, vol.~12, pp. 34\,656--34\,667, Mar. 2024.

\bibitem{che2021does}
D.~Che, J.~Cho, and X.~Chen, ``Does probabilistic constellation shaping benefit {IM-DD} systems without optical amplifiers?'' \emph{Journal of Lightwave Technology}, vol.~39, no.~15, pp. 4997--5007, May 2021.

\bibitem{seimetz2009high}
M.~Seimetz, \emph{High-Order Modulation for Optical Fiber Transmission}.\hskip 1em plus 0.5em minus 0.4em\relax Springer, 2009.

\bibitem{hui2019introduction}
R.~Hui, \emph{{Introduction to Fiber-Optic Communications}}.\hskip 1em plus 0.5em minus 0.4em\relax Academic Press, 2019.

\bibitem{napoli2017digital}
A.~Napoli, M.~M. Mezghanni, S.~Calabro, R.~Palmer, G.~Saathoff, and B.~Spinnler, ``Digital predistortion techniques for finite extinction ratio {IQ} {M}ach--{Z}ehnder modulators,'' \emph{Journal of Lightwave Technology}, vol.~35, no.~19, pp. 4289--4296, Oct. 2017.

\bibitem{soriaga2007determining}
J.~B. Soriaga, H.~D. Pfister, and P.~H. Siegel, ``Determining and approaching achievable rates of binary intersymbol interference channels using multistage decoding,'' \emph{IEEE Transactions on Information Theory}, vol.~53, no.~4, pp. 1416--1429, Apr. 2007.

\bibitem{Kaplan1993_RatesExponentsCompondChan}
G.~Kaplan and S.~{Shamai (Shitz)}, ``Information rates and error exponents of compound channels with application to antipodal signaling in a fading environment,'' \emph{A\"{E}U}, vol.~47, no.~4, pp. 228--239, 1993.

\bibitem{Merhav1994_InfoRatesMismDecod}
N.~{Merhav}, G.~{Kaplan}, A.~{Lapidoth}, and S.~{Shamai (Shitz)}, ``On information rates for mismatched decoders,'' \emph{IEEE Transactions on Information Theory}, vol.~40, no.~6, pp. 1953--1967, Nov. 1994.

\bibitem{Martinez2009_TransIT_bicm}
A.~{Martinez}, {Guill\'{e}n i F\`{a}bregas}, G.~{Caire}, and F.~M.~J. {Willems}, ``Bit-interleaved coded modulation revisited: A mismatched decoding perspective,'' \emph{IEEE Transactions on Information Theory}, vol.~55, no.~6, pp. 2756--2765, June 2009.

\bibitem{Szczecinski2015_BICMbook}
L.~Szczecinski and A.~Alvarado, \emph{Bit-Interleaved Coded Modulation: Fundamentals, Analysis, and Design}.\hskip 1em plus 0.5em minus 0.4em\relax Chichester, UK: John Wiley \& Sons, 2015.

\bibitem{scarlett2014mismatched}
J.~Scarlett, A.~Martinez, and A.~G. i~F{\`a}bregas, ``Mismatched decoding: Error exponents, second-order rates and saddlepoint approximations,'' \emph{IEEE Transactions on Information Theory}, vol.~60, no.~5, pp. 2647--2666, May 2014.

\bibitem{pang2020}
X.~Pang, O.~Ozolins, R.~Lin, L.~Zhang, A.~Udalcovs, L.~Xue, R.~Schatz, U.~Westergren, S.~Xiao, W.~Hu \emph{et~al.}, ``200 {Gbps/lane IM/DD} technologies for short reach optical interconnects,'' \emph{Journal of Lightwave Technology}, vol.~38, no.~2, pp. 492--503, Jan. 2020.

\bibitem{Alvarado12b}
A.~Alvarado, F.~Br\"{a}nnstr\"{o}m, E.~Agrell, and T.~Koch, ``High-{SNR} asymptotics of mutual information for discrete constellations with applications to {BICM},'' \emph{IEEE Transactions on Information Theory}, vol.~60, no.~2, pp. 1061--1076, Feb. 2014.

\end{thebibliography}

\end{document}